\documentclass[letterpaper, 10 pt, conference]{ieeeconf}
\usepackage{amssymb, amsmath}
\usepackage{xcolor,comment}
\usepackage{graphics,eepic,epic}
\usepackage{cite}
\usepackage{flushend}
\usepackage{tkz-euclide}
\usetkzobj{all}
\usepackage{tikz}
\title{Coding for Fast Content Download}
\author{
\authorblockN{Gauri Joshi}
\authorblockA{EECS Dept.,
MIT\\
Cambridge, MA 02139, USA \\
Email: gauri@mit.edu}
\and
\authorblockN{Yanpei Liu}
\authorblockA{ECE Dept.,
Univ.~Wisconsin Madison\\
Madison, WI, 53705, USA \\
Email: yliu73@wisc.edu}
\and
\authorblockN{Emina Soljanin}
\authorblockA{
Bell Labs, Alcatel-Lucent \\
Murray Hill NJ 07974, USA\\
Email: emina@bell-labs.com}
}

\newtheorem{defn}{Definition}
\newtheorem{rem}{Remark}
\newtheorem{lem}{Lemma}
\newtheorem{thm}{Theorem}

\newcommand{\expec}{\text{E}}
\newcommand{\var}{\text{V}}

\DeclareMathOperator*{\argmin}{arg\,min}
\begin{document}
\maketitle
\begin{abstract}
%(Added by Gauri) {\color{red} I like this abstract -- Yanpei}
%In this paper 
We study the fundamental trade-off between storage and content download time. We show that the download time can be significantly reduced by dividing the content into chunks, encoding it to add redundancy and then distributing it across multiple disks.
%For the same total storage used, the use of coding to add redundancy provides finer granularity than simple replication, and hence is more effective in exploiting the diversity in storage to reduce the download time.
We determine the download time for two content access models -- the fountain and fork-join models that involve simultaneous content access, and individual access from enqueued user requests respectively. For the fountain model we explicitly characterize the download time, while in the fork-join model we derive the upper and lower bounds. Our results show that coding reduces download time, through the diversity of distributing the data across more disks, even for the total storage used. %This delay-storage trade-off provides insight into practical system design.
%We show that by employing erasure correcting codes, one can divide, encode and distribute the content across multiple disks for redundancy in storage which significantly reduces the content retrieval delay. Given the same storage, erasure correcting codes provide a more flexible way to exploit the benefit of redundancy storage at a finer granularity when compared against the storage solution based on duplication. Formulated as a fork-join system, we analytically study the delay performance under different practical models with various redundancy levels. For the fountain model we explicitly characterize the mean response time while in the queuing model we derive the upper and lower bound. Our simulation results also provide insight in practical system design.
\end{abstract}


\section{Introduction}
Consumers of cloud storage and content centric networking demand that their content be reliably stored and quickly accessible. Cloud storage providers today strive to meet both demands by simply replicating content throughout the storage network over multiple disks. A large body of recent literature proposes erasure coding as a more efficient way to provide reliability.
%{\bf Shouldn't the following sentences go to paragraph in the intro about content accessibility? -- Gauri} {\color{red} Your call -- Yanpei}

Research in coding for distributed storage was galvanized by the results reported in \cite{dimakis07}. Prior to that work, literature on distributed storage recognized that when compared with replication, coding can offer huge storage savings for the same reliability levels. But it was also argued that the benefits of coding are limited, and can easily be outweighed by certain disadvantages and extra complexity. Namely, to provide reliability in multi-disk storage systems, when some disks fail, it must be possible to restore either the exact lost data or an equivalent reliability with minimal download from the remaining storage. The cost of this repair regeneration was considered much higher in coded than in replication system \cite{liskov}, until \cite{dimakis07} established existence and advantages of new regenerating codes. This work was then quickly followed, and the area is very active today (see e.g., \cite{kumar, kumar2} and references therein). 

%Namely, to provide reliability in multi-disk storage systems, when some disks fail, it must be possible to recover data from the remaining storage, and in addition the failed disks have to be {\bf replaced to restore either the exact lost data or an equivalent reliability with minimal download from the remaining storage. (Modified -- Yanpei).} The cost of this repair (regeneration) was considered much higher in coded than in replication system \cite{liskov}, until \cite{dimakis07} established existence and advantages of new regenerating codes. This work

A related line of work is concerned with another potential weakness of coding in distributed storage. Namely, if any part of the data changes, the corresponding coded packets must be updated accordingly. To minimize the cost of such updates, the authors in \cite{ankitISIT11} propose a class of randomized codes which have update complexity 
scaling logarithmically with the size of data but can correct a linearly scaled number of disk failures. Furthermore, the existence of such update efficient codes that also minimize the repair bandwidth 
for exact data reconstruction was established in \cite{ankitISIT11}.

Content accessibility is another main property of interest. In current multi-disk, cloud storage systems (e.g., Amazon), content files stored on the same disks may be simultaneously requested by multiple users. The file accessibility, therefore, depends on the dynamics of requests, and is limited by the disks' I/O bandwidth. In practice, it is again commonly improved by replicating content on multiple disks, which in turn requires more energy. Only recently was it realized that erasure coding can guarantee the same level of content accessibility with lower storage than replication. \cite{berk_isit,ulric_muriel_emina}. In \cite{ulric_muriel_emina}, it is considered that when there are multiple access requests, all but one of them are blocked, and the accessbility is measured in terms of blocking probability. In \cite{berk_isit}, multiple requests are placed in a queue instead of blocking. The authors propose a scheduling scheme to map requests to servers (or disks) to minimize the waiting time.

%These papers are concerned with the trade-off between accessibility and storage in multi-disk systems. %{\bf This paragraph needs to be shortened, maybe merged with the previous para?. -- Gauri} {\color{red} Your call -- Yanpei} When requests that cannot be served upon arrival are blocked, the performance is measured by blocking probability. This case has been considered in \cite{ulric_muriel_emina} for some storage area network scenarios commonly used in practice. It has been shown that network coding can be used to guarantee the same performance in terms of blocking (e.g., by some industry standard) as replication but with much fewer disks and thus much less energy.

In this paper, we assume that requests that cannot be served upon arrival are queued, but we measure the accessibility in terms of the download time which includes the waiting time in queues and the time taken to read the data from the disk, which could be random. When the content is available redundantly on multiple disks, it is sufficient to read it only from a subset of these disks in order to recover the data. The key contribution of our work is to analyze how waiting for only a subset of disks to be read, provides diversity in storage which helps achieve a significant reduction in the download time.
%(Edited- Gauri) {\color{red} This last sentence looks strange to me. Do you want to say something like ``The contribution of our work is the study on how much diversity storage can reduce the download time?" -- Yanpei}

% We assume that content of interest is available redundantly on multiple disks, or {\bf on multiple nodes throughout the network, entirely or in chunks. (Modified -- Yanpei)} Such scenarios may be a consequence of caching throughout a network or as a result of purposeful storage in data centers and storage area networks. Our idea is to also make redundant requests for content in order to improve accessibility through route diversity and load balancing.

Using redundancy in coding for delay reduction has also been studied in packet transmission
\cite{kabatiansky_krouk_semenov, maxemchuk2, mia}, and in some other scenarios of content retrieval in \cite{Allerton10e}. Although they share some common spirit, they do not consider storage systems and the impact of redundancy coding in such scenarios.

This paper is organized as follows. In Section~\ref{sec:access_models} we introduce the two content access models investigated in this work. In Section~\ref{sec:main_idea} we present the central idea of how coding gives diversity and reduces download time.  Then we determine the download time for the two models in Section~\ref{sec:multiple_fountain} and Section~\ref{sec:fork_join} respectively. Finally we conclude in Section~\ref{sec.conclusion}.
%(Added -- Yanpei, Edited- Gauri)} {\color{red} Looks good -- Yanpei}

\section{Two Content Access/Delivery Models}
\label{sec:access_models}
We consider two specific content access models in order to isolate and emphasize different sources of delay in content delivery, and show how they can be addressed through coding. In general, a content delivery model is some hybrid of the two models considered here.
% -- (Edited -Gauri)} {\color{red} Looks good -- Yanpei} %our two.

\subsection{The Fountain Model}
\label{sec.fountain_model}
Our first model can describe, e.g., a content delivery network scenario, where content may not be available at the point of request, and there is a delay associated with waiting for the content to become available at the contacted server. Once in possession of the content, the server can deliver it to the users by, e.g., broadcast enabled by digital fountain codes \cite{raptor_codes}. Thus multiple users do not affect each other's delivery time. Another scenario that can be modeled in this way is when content is broadcast at some prescribed times, but the arrival of users is random. We will refer to this model as the {\it fountain} model (where fountains are turned on at random times).

When a content request arrives at the server,  the content has to be fetched from the distribution network and then transmitted to the user. The waiting time to obtain the content from the network is a non-negative random variable $W$. Once the content is obtained, the time to deliver it to the user is a positive random variable $D$, which is proportional to the size of content and erasure rate of the channel. %We assume a perfect channel in this paper. Thus, $D$ is {\bf fixed (``fixed random variable" means ``same random variable" in my opinion. But I think in the later section it is actually a constant. So I suggest we explicitly say ``constant" -- Yanpei)} across requests and proportional to the content size, which we here assume fixed across requests.
Multiple users can access a server simultaneously and the content is broadcast to these users. Hence, the response time of the system (the download completion time) is $W+D$, and is independent of the number of requests being served simultaneously.

\subsection{The Queueing Model}
%{\bf Shouldn't we replace the word 'server' by 'disk' here, to be consistent with the fork-join section, and the Allerton PPT? -- Gauri } {\color{red} Your call -- Yanpei}
Our second model can describe, e.g., a storage area network (SAN), where content is stored on a disk, which can be accessed only by one user at a given time. The delay in this model is associated with the response time of the queueing systems. In this model, multiple requests by the same user do affect each other's content download time. We will refer to this model as the {\it queueing} model.

%Our second model can describe, e.g., a storage area network, where content is available at the point of request, but the server is busy serving other users, and requests have to be queued. The delay in this model is associated with the response time of the queueing systems. In this model, multiple users or multiple requests by the same user do affect each other's content delivery time. We will refer to this model as the {\it queueing} model.

When a content request arrives at the disk, it enters a first-come-first-serve queue. After a request reaches the head of the queue, it takes some random service time to read the content from the disk. We model this service time a random variable with mean $1/\mu$. Here again, the download time is the sum of two components: the waiting time in queue and the service time required to read from the disk.
% (Added -- Yanpei, Edited--Gauri)} {\color{red} Looks good -- Yanpei}

%{\bf When a content request arrives at the server, it enters a first-come first-serve queue. After a request reaches the head of the queue, it takes some random time to complete the download. We model this download completion time a random variable with mean $1/\mu$. Thus the waiting time to obtain the content from the server is the sum of two components: the time spent in queuing and the time spent in downloading the content from the server. (Added -- Yanpei, Edited--Gauri)}

\section{Reducing Delay by Coding}
\label{sec:main_idea}
In both of our models, the download completion time is a random variable. One natural way to reduce this time is to replicate the content across $n$ independent servers (or disks). Then if the user issues $n$ requests, one to each of the $n$ servers, it only needs to wait for the one of the requests to be served. This strategy gives a sharp reduction in download time, but at the cost of $n$ times more storage space and the cost of processing multiple requests.
%(Edited - Gauri)} {\color{red} Looks good -- Yanpei}

%In both of our models, the download completion time is a random variable. One natural way to reduce this time is for the user to issue multiple requests for content and wait for the first to be fulfilled. With this strategy, a reduction in response time is possible only if there is a certain degree of independence in response times for each of the multiple requests. That in general requires that content be entirely replicated on multiple servers, which in turn means higher storage and other costs.

We thus argue that it is more economical to divide the content into $k$ blocks, encode them into $n>k$ coded blocks, and store them on $n$ different servers (one block per server). Each incoming request is sent to all $n$ servers, and the content can be recovered when any $k$ out of $n$ blocks are successfully downloaded.

%{\bf Edited Gauri -- The sentences here were not clear. Is there a class of codes called 'parity-check'? And repetition codes are a special case that works only for $k=1$} {\color{red} That sentence is just to remind the reader that the only binary MDS codes are trivial codes (repetition codes and parity codes) \cite{Roth}. However for non-binary codes, there is a large group of such codes (e.g., Reed-Solomon). I think this should be added back on -- Yanpei} {\bf Okay, I agree - Gauri}
This can be achieved  by using an $(n,k)$ maximum distance separable (MDS) code to encode the $k$ blocks of content. 
%The only binary codes that have this property are parity check codes and repetition codes. However for non-binary codes, there is a large group of such codes (e.g., Reed-Solomon codes). 
MDS codes have been suggested to provide reliability against server outages (or disk failures). In this paper we show that, in addition to error-correction, we can exploit these codes to reduce the download time of the content.

%{\bf This can be achieved using an MDS erasure correcting code (see for example \cite{berkelamp}) to treat the rest $n-k$ blocks as ``erased". The only binary codes that have this property are parity check codes and repetition codes. However for non-binary codes, there is a large group of such codes (e.g., Reed-Solomon codes). (Added -- Yanpei).} When the response times of the $n$ servers are independent, the overall response time is the $k^{th}$ order statistic of the individual response times. Introducing redundancy in this way clearly protects the content against server outages. In this paper we show that, in addition to using it for error-correction, we can exploit the coding redundancy to reduce the expected download time of the content. Note that for the fountain model, the response times of the $n$ servers are independent.

%Note that for the fountain model, the response times of the $n$ servers are independent.
%{\bf When we have an array of $n$ disks, the content can be divided into $k<n$ blocks, encoded into $n$ blocks and stored on the disks. A content request is sent to each of the $n$ queues, and the content download is complete when any $k$ out of $n$ queues finish serving the request. We define this system as a $(n,k)$ fork-join system in Section~\ref{sec:fork_join}. In this paper we show that the diversity obtained by spreading the content across several disks reduces the expected download time significantly. (This part somewhat repeats the previous paragraph -- Yanpei)}.

Note that for the fountain model, since multiple users can simultaneously access the content,  the response times (waiting plus delivery time) of the $n$ servers are independent. The download time is the $k^{th}$ order statistic of the response time of each server. However, the analysis of download time for the queueing model is more challenging because the the response times of the $n$ queues are not independent.

%We now  {\bf provide some background (Modified -- Yanpei)} on order statistics of exponential random variables.  These results are used in the analysis of response time in the following sections. {\bf For the complete treatment in order statistics, please refer to \cite{order_stat}. (Added -- Yanpei)}

Since in both models we require the first $k$ out of $n$ blocks to be downloaded, we now provide some background of finding the $k^{th}$ order statistic of $n$ independent and identically distributed (i.i.d) random variables. For a more complete treatment in order statistics, please refer to \cite{order_stat}.  %(Edited-Gauri)} {\color{red} Looks good -- Yanpei}

Let $X_1, \,\, X_2, \,\, \cdots \, X_n$ be i.i.d.\ random variables. Then, $X_{k,n}$, the $k^{th}$ order statistic of $X_i$ , $ 1 \leq i \leq n$ , or the $k^{th}$ smallest variable has the distribution,
\begin{align}
f_{X_{k,n}}(x) &= n \binom{n-1}{k-1} F_{X}(x)^{k-1}(1-F_{X}(x))^{n-k}f_{X}(x)
\end{align}
where $F_X(x)$ and $f_X(x)$ and the distribution and density functions of $X$ respectively. In particular, if $X_i$'s are exponential with mean $1/\mu$, then the expectation and variance of order statistic $X_{k,n}$ are given by,
\begin{align}
\label{eq.orderstat_mean}
\expec[X_{k,n}] &= \frac{1}{\mu}\sum_{i=1}^{k} {\frac{1}{n-k+i}} = \frac{1}{\mu} (H_n - H_{n-k}) \\
\label{eq.orderstat_var}
\var[X_{k,n}] &= \frac{1}{\mu^2} \sum_{i=1}^{k} {\frac{1}{(n-k+i)^2}} = \frac{1}{\mu^2} (H_{n^2} - H_{(n-k)^2}),
\end{align}
where $H_n$ and $H_{n^2}$ are generalized harmonic numbers defined as
\begin{align}
\label{eq.harmonic_no}
H_n =  \sum_{j=1}^n \frac{1}{j} ~~ \text{and} ~~ H_{n^2} = \sum_{j=1}^n \frac{1}{j^2}.
\end{align}

Note that $\expec[X_{k,n}]$ decreases as $k$ decreases for a given $n$. This fact will provide us some intuition to understand the analysis of download time for the fountain and queueing models presented in Section~\ref{sec:multiple_fountain} and Section~\ref{sec:fork_join} respectively. %Furthermore, as we will study in the following sections, a decreased $k$ leads to a larger delivery time in the fountain model and more storage in the queueing model.

%\textcolor{red}{In this section, we should be careful about "independence". Please check if the work has been appropriately used or appropriately omitted.} {(\bf What ``independence"? -- Yanpei) (I think the use of independence is clear now -- Gauri)  }

\section{Multiple Fountains}
\label{sec:multiple_fountain}

In this section we investigate the redundancy storage in the context of fountain model. Content requests (
e.g., request for videos, news or other information) from customers are sent to a network of servers. 
%The servers respond to the requests by delivering the requested content. 
We focus in particular on {\em multiple fountain} content retrieval systems defined as follows.
\begin{defn}[$(n, k)$ multiple fountain]
An $(n, k)$ multiple fountain content retrieval system contains $n$ servers. 
Every content request entering the system is forked to $n$ servers. Requests are served as soon as the
content becomes available, which happens at a random time independently of request arrivals.
A request is satisfied when any $k$ out of $n$ servers have responded and delivered their messages.
\end{defn}

%{\bf Edited by Gauri - Trying to keep the names consistent - total time is download time which equals waiting plus delivery time} {\color{red} Looks good -- Yanpei}
Recall that the content is divided into $k$ blocks and encoded into $n > k$ coded blocks which are stored on $n$ servers 
(one block per server). Content is said to be downloaded when any $k$ out of $n$ blocks are successfully delivered to the user. 
The fountain model described in Section~\ref{sec.fountain_model} assumes that multiple users can access the server simultaneously 
(i.e, no queueing). The response time for each server is the sum of waiting time for the content to become available and the time taken 
to deliver each content block. We model the waiting time $W$ as an exponentially distributed random variable with mean $1/\mu$. 
After the content becomes available, the server delivers it to the customer in constant time $\frac{D}{k}$, 
where the factor $1/k$ appears because each server only delivers $1/k$ fraction of the content.

The mean response time, i.e., the time taken to download the content, from the $(n, k)$ multiple fountain system is given by the following theorem.
\begin{thm}
\label{tm.MFMainTheorem}
The mean response time $T_{(n,k)}$ of a content retrieval system is
\begin{align}
\label{eq.MFMeanResponseTime}
T_{(n,k)} = \frac{1}{\mu}(H_n - H_{n-k}) + \frac{D}{k},
\end{align}
where $H_n$ is defined in (\ref{eq.harmonic_no}).
\end{thm}
\begin{proof}
Since each message is delivered to the customer in constant time $\frac{D}{k}$, the mean response time for a request equals the waiting time until the content becomes available at $k$ servers plus the delivery time $\frac{D}{k}$. The expected waiting time is the $k^{th}$ order statistic of $n$ i.i.d.~exponential random variables with mean $1/\mu$, which  has mean $\frac{1}{\mu}(H_n - H_{n-k})$ (c.f.~(\ref{eq.orderstat_mean})). %Then adding the download time we obtain the result.
\end{proof}

We notice that it is possible to have an optimal $k$ such that (\ref{eq.MFMeanResponseTime}) is minimized. The intuition behind this is the trade-off between the waiting time $\frac{1}{\mu}(H_n - H_{n-k})$ and the content delivery time $\frac{D}{k}$, as $k$ varies from $1$ to $n$. When $k$ is small, the $T_{(n,k)}$ is dominated by the delivery time $\frac{D}{k}$. But as $k$ increases, $T_{(n,k)}$ is dominated by the waiting time due to the increase in $\frac{1}{\mu}(H_n - H_{n-k})$, and decrease in $\frac{D}{k}$. The following lemma gives the optimal value of $k$.

\begin{lem}
The $k$ that minimizes (\ref{eq.MFMeanResponseTime}) is given by
\begin{align}
\label{eq.MFOptimalk}
k^* = \argmin_k T_{(n,k)} \approx
\left \lceil \frac{-D \mu+ \sqrt{D^2\mu^2 + 4nD\mu}}{2} \right \rceil.
\end{align}
\end{lem}
\begin{proof}
We use the log approximation for $H_n$, i.e., $H_n \approx \log (n) + O(1)$ and $H_{n-k} \approx \log (n-k) + O(1)$. Substitute in (\ref{eq.MFMeanResponseTime}) we obtain
\begin{align*}
T_{(n,k)} = \frac{1}{\mu}\log \left (\frac{n}{n-k} \right ) + \frac{D}{k}.
\end{align*}
Taking derivative with respect to $k$ and set it to $0$ we obtain
\begin{align*}
\frac{1}{\mu}k^2 + Dk - Dn = 0,
\end{align*}
which has root $k^* = \frac{-D \mu + \sqrt{D^2 \mu^2+ 4nD\mu}}{2}$ as in (\ref{eq.MFOptimalk}).
%{\bf Omitted the part about $1 \leq k^* \leq n$ -- Gauri} {\color{red} OK -- Yanpei}
%
%Note that it is always the case that
%\begin{align*}
%0 < k^* < n.
%\end{align*}
\end{proof}
Fig.~\ref{fig.sim_multiple_Fountain} shows the mean response time 
$T_{10,k}$ versus $k$ for the $(10, k)$ multiple fountain system with parameters delivery time $D =5$ and various values of mean waiting time $1/\mu$. 
We observe that the optimal value of $k^*$ increases with $1/\mu$.
\begin{figure}[hbt]
	\includegraphics[scale=0.6]{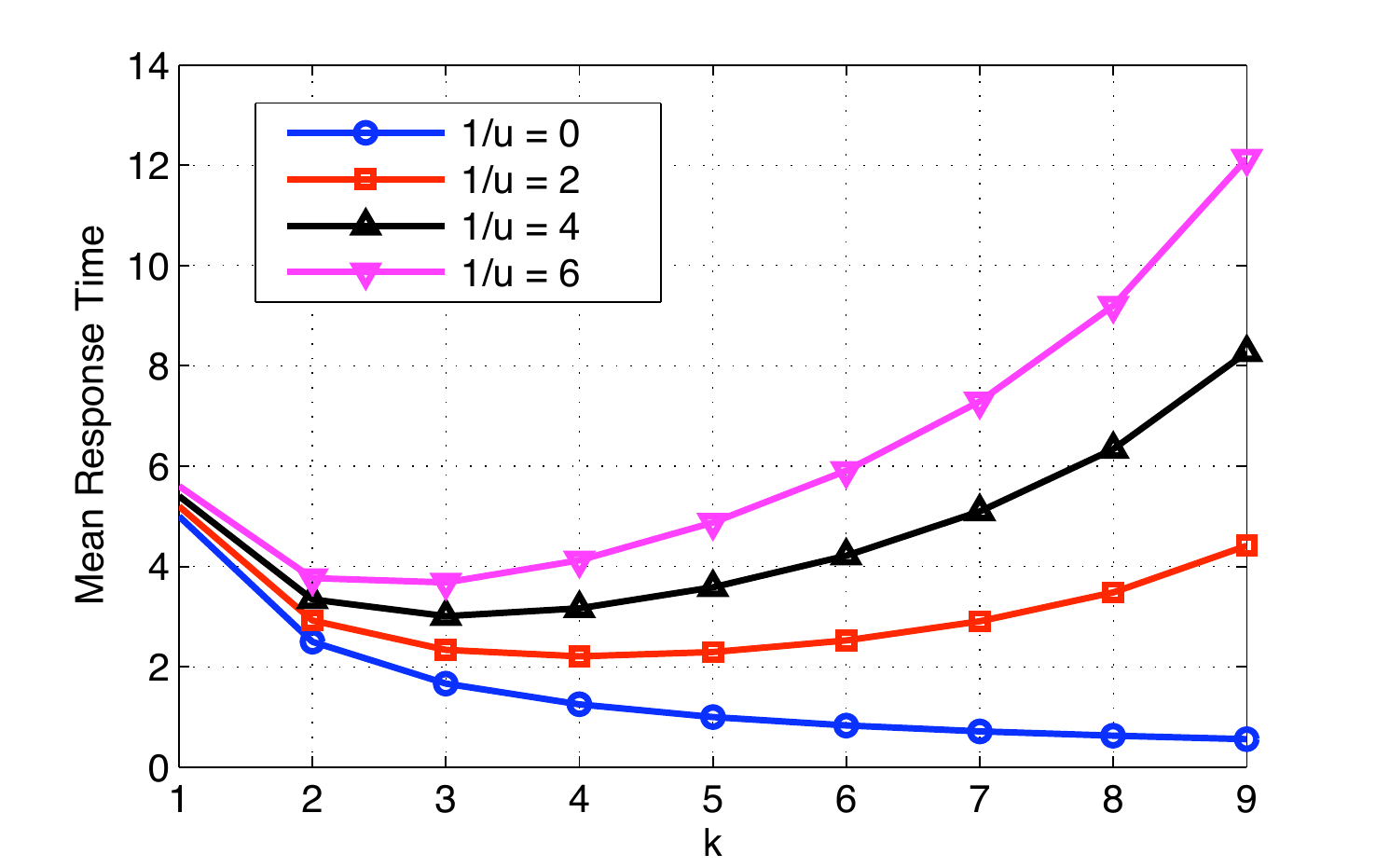}
	\caption{Mean response time simulation for $(10, k)$ multiple fountain system with parameters $\frac{1}{\mu} = 0, 2, 4, 6$ and $D =5$. Note that $1/\mu = 0$ means that the content is immediately available thus the download time decreases as $k$ increases. On the contrary, when $D=0$ (not shown in this plot), the download time goes down as $k$ decreases. }
	\label{fig.sim_multiple_Fountain}
\end{figure}
\begin{comment}
Simulation results for $(n, k)$ multiple fountain systems with different $W$ and $D = 0$ are shown in Fig.~\ref{fig.sim_waiting_time}.

\begin{figure}[!t]
	%\pdfembed pages 1 2 3 {CDexamples.pdf}
	\includegraphics[scale=0.32, page = 3]{CDexamples.pdf}
	\includegraphics[scale=0.32, page = 1]{CDexamples.pdf}
	\includegraphics[scale=0.32, page = 2]{CDexamples.pdf}
	\caption{Simulation results for $(n, k)$ multiple fountain systems with different distribution of waiting time and $D = 0$.}
	\label{fig.sim_waiting_time}
\end{figure}
\end{comment}

\section{Fork-Join Queues}
\label{sec:fork_join}

%---------BEGIN ----  Added by Gauri ------------------------------
In this section we consider the second content delivery model described in Section~\ref{sec:main_idea}, the queueing model. We show how coding can help minimize the time taken to download time of a content which is stored on an array of disks. We refer to this time as the response time. Although we focus on this storage model, it is possible to extend our results to other distributed systems such as parallel cluster computing \cite{map_reduce}.

\subsection{System Model}
\label{subsec:sys_model}
Consider that a content $F$ of unit size, divided into $k$ blocks of equal size. It is encoded to $n>k$ blocks using a $(n,k)$ maximum distance separable (MDS) code, and the coded blocks are stored on an array of $n$ disks. MDS codes have the property that any $k$ out of the $n$ blocks are sufficient to reconstruct the entire file.
%\begin{comment}
An illustrative example with $n=3$ disks and $k=2$ is shown in Fig.~\ref{fig:sys_model_example}. The content $F$ is split into equal blocks $a$ and $b$, and stored on $3$ disks as $a$, $b$, and $a \oplus b$, the exclusive-or of blocks $a$ and $b$. Thus each disk stores content of half the size of file $F$. Downloads from any $2$ disks jointly enable reconstruction of $F$.
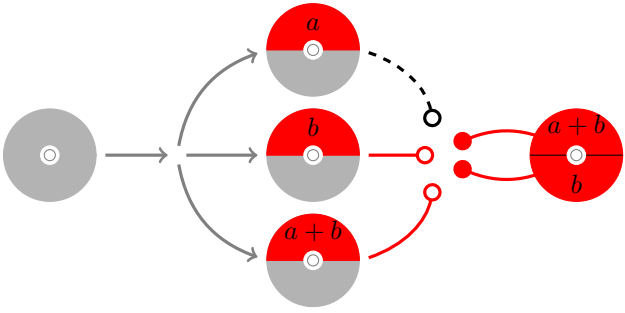
\begin{figure}[hbt]
\label{fig:sys_model_example}
\begin{center}
\begin{tikzpicture} [scale=0.7]
\def\myangle{180};
\def\dr{25pt}
\def\hr{3pt}
\def\tr{5pt}
\coordinate (A) at (-5,2);
%\node[right] at (-5,5.25) (fji) {user requests are forked to $3$ disks, each containing half };
\draw [black!30, - , fill] (A) -- ($(A)+(\dr,0)$)  arc (0:360:\dr) -- cycle;
\draw [white, fill] (A) circle (\tr);
\draw [black!50] (A) circle (\hr);
\node at ($(A)+(\dr,0)$)(I) {};
\node at ($(A)+(2.75*\dr,0)$)(IF) {};
\foreach \n in {0,2,4} {
\coordinate (A) at (0,\n);
\node at ($(A)-(\dr,0)$)(F-\n) {};
\node at ($(A)+(\dr,0)$)(J-\n) {};
\draw [black!30, - , fill] (A) -- ($(A)+(\dr,0)$)  arc (0:360:\dr) -- cycle;
\draw [red, - , fill] (A) -- ($(A)+(\dr,0)$)  arc (0:\myangle:\dr) -- cycle;
\draw [white, fill] (A) circle (\tr);
\draw [black!50] (A) circle (\hr);
}
\path (I) edge[->, gray, very thick] (IF);
\path (IF) edge[->, bend right, gray, very thick] (F-0);
\path (IF) edge[->, bend left, gray, very thick] (F-4);
\path (IF) edge[->, gray, very thick] (F-2);
\draw ($(0,0)+(0,1.1*\tr)$) node [above] {{$a+b$}};
\draw ($(0,2)+(0,1.1*\tr)$) node [above] {{$b$}};
\draw ($(0,4)+(0,1.1*\tr)$) node [above] {{$a$}};
\coordinate (A) at (5,2);
\draw [red, - , fill] (A) -- ($(A)+(\dr,0)$)  arc (0:360:\dr) -- cycle;
\draw [black] (A) -- ($(A)+(\dr,0)$);
\draw [black] (A) -- ($(A)-(\dr,0)$);
\draw [white, fill] (A) circle (\tr);
\draw [black!50] (A) circle (\hr);
\node[text width=0.15cm, text height=0.5cm] at ($(A)-(2.75*\dr,0)$)(O) {};
\node[text width=0.2cm,text height=0.75cm] at ($(A)-(3*\dr,0)$)(OL) {};
\node at ($(A)-(0.9\dr,0)$)(OJ) {};
\path (OJ) edge[<-*, bend right, red,very thick] (OL);
\path (OJ) edge[<-*, bend left, red,very thick] (OL);
\path (J-0) edge[-o, bend right, red, very thick] (O.south west);
\path (J-4) edge[-o, bend left,very thick,dashed] (O.north west);
\path (J-2) edge[-o, red,very thick] (O);
\draw ($(5,2)+(0,1.1*\tr)$) node [above] {{$a+b$}};
\draw ($(5,2)-(0,1.1*\tr)$) node [below] {{$b$}};
\end{tikzpicture}
\end{center}
\caption{ A $(3, 2)$ fork-join system; storage is $50\%$ higher, but response time (per disk \& overall) is reduced. \label{fig:sys_model_example}}
\end{figure}
%
%
%\end{comment}
Each user's request for content $F$ is forked to all the $n$ disks. Our objective is to determine the mean response time of the system -- the expected time from the arrival of a request until it finishes service by reading the content from some $k$ of the $n$ disks. In order to evaluate the response time we model it as an $n$-fork $k$-join system which is defined as follows.

\begin{defn}[$(n,k)$ fork-join system]
\label{defn:n_fork_k_join}
An $(n,k)$ fork-join system consists of $n$ processing nodes (fork nodes). Every arriving job is divided into $n$ tasks which are sent
to the queue at each of the $n$ nodes. A task is served when it arrives at the top of its queue. The job departs the system when any $k$ out of $n$ tasks are served by their respective nodes. The remaining $n-k$ tasks abandon their queues and exit the system without receiving service.
\end{defn}
The $(n,n)$ fork-join system, known in literature as fork-join queue, has been extensively studied in, e.g., \cite{kim_agarwal, nelson_tantawi, varki_merc_chen}. However, the $(n,k)$ generalization in Definition~\ref{defn:n_fork_k_join} above has not been previously studied to our best knowledge.

We consider an $(n,k)$ fork-join system where each node represents a disk from which content is being downloaded. Download requests arrive according to a Poisson process with rate $\lambda$ per second. Every request is sent to each of the $n$ disks. The time taken to download one unit of data is exponential with mean $1/\mu$. Since, each disk is requested to provide $1/k$ units of data, the service time for each node is exponentially distributed with mean $1/\mu'$ where $\mu' =  k\mu$. Define the  load factor $\rho \triangleq \lambda/\mu'$.  We assume $\rho > 1$, or equivalently $\mu' > \lambda$ to ensure stability of the queue at each fork node.

%Figure~\ref{fig:fork_join_queue} illustrates a $(5,3)$ fork-join system.

%---------END ----  Added by Gauri - 24th July -----------------------------

%\begin{itemize}
%  \item Content $F$ is split into equal parts $a$ and $b$, and stored on $3$ disks\\
%  as $a$, $b$, and $a+b$ $\Rightarrow$ each disk stores half the size of $F$.
%  \item User request for $F$ are \textcolor{blue}{forked} to all $3$ disks.
%  \item Downloads from any $2$ disks \textcolor{blue}{jointly} enables reconstruction of $F$.
%\end{itemize}

\subsection{Bounds on Mean Response Time}
\label{subsec:bounds}

%---------BEGIN ----  Added by Gauri - 20th July -----------------------------

Our objective is to evaluate the mean response time $T_{(n,k)}$ of the $(n,k)$ fork-join system described in Section~\ref{subsec:sys_model}. It is the time from the arrival of a job until $k$ out of $n$ of its tasks are served by their respective nodes. %{\bf The mean response time $T_{(n,n)}$ of the $(n,n)$ fork-join system has been studied extensively in previous work, but the $(n,k)$ system with $k<n$ case has not been considered to our best knowledge. (It is a repeat of 2nd paragraph of subsection A above. Consider taking it out -- Yanpei)}

Even for the $(n,n)$ system, the mean response time has not been found in closed form -- only bounds are known. An exact expression for the response time is found only for the $(2,2)$ fork-join in \cite{nelson_tantawi}. The reason why the fork-join system is harder to analyze than a set of parallel independent $M/M/1$ queues is that each incoming job is sent to the $n$ queues. Hence, the arrivals to the queues are perfectly synchronized and the response times of the $n$ queues are correlated.

The simplest case of an $(n,k)$ fork-join system is the $(n,1)$ system. It is not hard to see that this system behaves exactly as an $M/M/1$ queue with arrival rate $\lambda$ and service rate $\mu'=n \mu$. Therefore its response time is exponential with the mean $T_{(n,1)}$ equal
to $1/(n\mu - \lambda)$. It is difficult to evaluate $T_{(n,k)}$ exactly for other cases, but the bounds we derive below are fairly tight.
%
\begin{comment}
{\bf (I take out the definition for generalized harmonic number as it no longer fits in here. The introduction of those harmonic numbers is put in Section III per Emina's comment -- Yanpei)}

\begin{defn}[Generalized Harmonic Number]
For a given sequence $a_i$, $ i \in \mathbb{N}$, the generalized harmonic number $H_{a_n}$ is defined as,
\begin{equation}
H_{a_n} = \frac{1}{a_1} + \frac{1}{a_2} + \cdots + \frac{1}{a_n}
\end{equation}
\end{defn}
\end{comment}

\begin{thm}[Upper Bound on Mean Response Time]
\label{thm:upper_bnd}
The mean response time $T_{(n,k)}$ of an $(n,k)$ fork-join system satisfies
\begin{align}
\label{eqn:upper_bnd}
T_{(n,k)} \leq  & \, \frac{H_n - H_{n-k}}{\mu'}\,  + \\
& \, \frac{\lambda\bigl [(H_{n^2} - H_{ (n-k)^2 }) + (H_n - H_{(n-k)})^2\bigr]}{2 \mu'^2 \bigl[ 1- \rho(H_n - H_{n-k})\bigr]} \nonumber
\end{align}
where $\lambda$ is the request arrival rate, $\mu' = k \mu$ is the service rate at each queue, $\rho = \lambda/\mu'$ is the load factor,
and the generalized harmonic numbers $H_n$ and $H_{n^2}$ are as given in (\ref{eq.harmonic_no}).
\end{thm}

\begin{proof}
We use a related, but easier to analyze queueing model called the split-merge system, to find this upper bound on $T_{(n,k)}$. In the $(n,k)$ fork-join queueing model, after a node serves one of the tasks, it is free to process the next task in its queue. On the contrary, in the split-merge model, all $n$ nodes are blocked until $k$ of them finish service. Thus, the job departs all the queues at the same time.
Since the nodes are not blocked in the fork-join system, the mean response time of the $(n,k)$ split-merge model is an upper bound on (and a pessimistic estimate of) $T_{(n,k)}$ for the $(n,k)$ fork-join system.

The $(n,k)$ split-merge system is equivalent to an $M/G/1$ queue where arrivals are Poisson with rate $\lambda$ and service time is a random variable $S$ distribution according to the $k^{th}$ order statistic of the exponential distribution.
% given by
%\[
%F_S(s) = n \binom{ n-1 }{k-1} \mu' e^{-\mu' (n-k+1)s } ( 1- e^{- \mu' s} )^k
%\]
The mean and variance of $S$ are (c.f.~(\ref{eq.orderstat_mean}) and (\ref{eq.orderstat_var}))
\begin{equation}
\expec[S] = \frac{H_n - H_{n-k}}{\mu'}  ~~ \text{and} ~~  \var[S] = \frac{H_{n^2} - H_{(n-k)^2}}{\mu'^2}.
\label{eqn:mvk_exp}
\end{equation}
The Pollaczek-Khinchin formula \cite{dsp_gallager} gives the mean response time $T$ of an $M/G/1$ queue in terms of the mean and variance of $S$ as follows.
\begin{equation}
T = \expec[S] + \frac{ \lambda \expec[S^2]} {2( 1-\lambda \expec[S])} \label{eqn:pollac_khin}
\end{equation}
where the second moment $\expec[S^2] = \var[S]+ \expec[S]^2$. Substituting the values of $\expec[S]$ and $\var[S]$ given by (\ref{eqn:mvk_exp}),
we get the upper bound (\ref{eqn:upper_bnd}).
\end{proof}

\begin{rem}%[Alternative approach to upper bound]
Note that the approach used in \cite{nelson_tantawi} to find an upper bound on the mean response time of the $(n,n)$ fork-join system cannot be extended to the $(n,k)$ fork-join system considered here. The authors in \cite{nelson_tantawi} show that the response times of the $n$ queues form a set of associated random variables \cite{assoc_rand_vars}. Associated random variables have the property that their expected maximum is less than that for independent variables with the same marginal distributions. Thus, the mean response time of the $(n,n)$ fork-join system is upper bounded by that of the system of $n$ independent $M/M/1$ queues. However, this property of associated variables does not hold for the $k^{th}$ order statistic for $k<n$.
\end{rem}

\begin{thm}[Lower Bound on Mean Response Time]
\label{thm:lower_bnd}
The mean response time $T_{(n,k)}$ of an $(n,k)$ fork-join queueing system satisfies
\begin{equation}
T_{(n,k)} \geq \frac{1}{\mu'} \bigl[ H_n - H_{n-k} + \rho (H_{n(n-\rho)} - H_{(n-k)(n-k-\rho)}) \bigr] \label{eqn:lower_bnd}
\end{equation}
where $\lambda$ is the request arrival rate, $\mu' = k \mu$ is the service rate at each queue, $\rho = \lambda/\mu'$ is the load factor,
and the generalized harmonic number $H_{n(n-\rho)}$ is given by
\begin{equation*}
H_{n(n-\rho)} = \sum_{j=1}^n \frac{1}{j(j-\rho)}.
\end{equation*}
\end{thm}

\begin{proof}
The lower bound in (\ref{eqn:lower_bnd}) is a generalization of the bound for the $(n,n)$ fork-join system derived in \cite{varki_merc_chen}. The bound for the $(n,n)$ system is derived by considering that a job goes through $n$ stages of processing. A job is said to be in the $j^{th}$ stage, for $0 \leq j \leq n-1$, if $j$ out of $n$ tasks have been served by their respective nodes and the remaining $n-j$ tasks are waiting to be served. The job will depart the system when all $n$ tasks are served.

For the $(n,k)$ fork-join system, since we only need $k$ tasks to finish service, the number of stages of processing is reduced. Each job now goes through $k$ stages of processing, where in the $j^{th}$ stage, for $0 \leq j \leq k-1$, $j$ tasks have finished processing and we are waiting for the $k-j$ more tasks to finish service in order to complete the job.

Consider two jobs $B_1$ and $B_2$ in the $i^{th}$ and $j^{th}$ stages of processing respectively.
Let $i>j$, or in other words, $B_1$ has completed more tasks than $B_2$. Since every incoming job is sent to all $n$ queues, this implies that $B_1$'s tasks will be in front $B_2$'s in all $n-i$ queues remaining to be served for $B_1$. Further, we can conclude that the mean service rate of job $B_2$ moving to the $(j+1)^{th}$ stage of processing is at most $(n-j)\mu'$. If the $n-j$ pending tasks are at the head of all the respective queues, then the service rate will be exactly $(n-j) \mu'$. However, $B_1$'s task could be ahead of $B_2$'s in one of the $n-j$ pending queues, due to which that task of $B_2$ cannot be immediately served. Hence, we have shown that for a job in the $j^{th}$ stage of processing, the mean service rate is at most $(n-j)\mu'$.

Consider an $M/M/1$ queue with arrival rate $\lambda$ and service rate $(n-j) \mu'$. Its response time is exponentially distributed with mean $T_j = 1/( (n-j)\mu' - \lambda)$. By the memory-less property of the exponential distribution, the total mean response time is the sum of the mean response times of each of the $k$ stages of processing, given by
\begin{align*}
T_{(n,k)} &\geq \sum_{j=0}^{k-1} \frac{1}{(n-j)\mu' - \lambda}
= \frac{1}{\mu'}  \sum_{j=0}^{k-1} \frac{1}{ (n-j) - \rho} \\
&= \frac{1}{\mu'}  \sum_{j=0}^{k-1} \Bigl[ \frac{1}{n-j} + \frac{\rho}{(n-j)(n-j-\rho)} \Bigr] \\
&= \frac{1}{\mu'} \bigl[ H_n - H_{n-k} + \rho \cdot (H_{n(n-\rho)} - H_{(n-k)(n-k-\rho)}) \bigr]
\end{align*}

\end{proof}

Hence, we have found lower and upper bounds on the mean response time $T_{(n,k)}$. In Section~\ref{subsec:fj_results}, we perform simulations to check the tightness of these bounds. These help us answer some practical questions in designing storage systems with the minimum download completion time.

%--------- Yanpei Added ----------------------------------------------
\subsection{Extension to General Service Time Distribution}
\label{subsec:fj_general_service_time}
In this section we derive the upper bound on expected download time with a general service time distribution at each node, instead of the exponential service time considered so far. Let $X_1, X_2, \ldots, X_n$ be the i.i.d random variables representing the service times of the $n$ nodes, with expectation $\expec[X_i] = \frac{1}{\mu'}$ and variance $\var[X_i] = \sigma^2$ for all $i$.

\begin{thm}[Upper Bound with General Service Time]
The mean response time $T_{n,k}$ of an $(n, k)$ fork-join system with general service time $X$ such that $\expec[X] = \frac{1}{\mu'}$ and $\var[X] = \sigma^2$ satisfies
\begin{align}
T_{(n,k)} \leq \frac{1}{\mu'} &+ \sigma \sqrt{\frac{k-1}{n-k+1}} \nonumber \\
&+ \frac{\lambda \left[ \left ( \frac{1}{\mu'} + \sigma \sqrt{\frac{k-1}{n-k+1}} \right)^2 + \sigma^2 C(n,k) \right]}{2 \left[ 1- \lambda \left( \frac{1}{\mu'} + \sigma \sqrt{\frac{k-1}{n-k+1}} \right ) \right ]} \label{eqn:upper_bnd_gen}.
\end{align}
\end{thm}
\begin{proof}
The proof follows from Theorem~\ref{thm:upper_bnd} where the upper bound can be calculated using $(n, k)$ split-merge system and Pollaczek-Khinchin formula (\ref{eqn:pollac_khin}). Unlike the exponential distribution, we do not have an exact expression for $S$, i.e., the $k^{th}$ order statistic of the service times $X_1, \,\, X_2, \,\, \cdots X_n$. Instead, we use the following upper bounds on the expectation and variance of $S$ derived in \cite{ArnoldGroeneveld} and\cite{Papadatos}.
\begin{align}
\expec[S] &\leq \frac{1}{\mu'} + \sigma \sqrt{\frac{k-1}{n-k+1}} \label{eqn:upperbound_exp} \\
\var[S] &\leq C(n,k) \sigma^2, \label{eqn:upperbound_var}
\end{align}

The proof of (\ref{eqn:upperbound_exp}) involves Jensen's inequality and Cauchy-Schwarz inequality. For details please refer to \cite{ArnoldGroeneveld}. The constant $C(n,k)$ depends only on $n$ and $k$, and can be found in the table in \cite{Papadatos}. Holding $n$ constant, $C(n,k)$ decreases as $k$ increases.  The proof of (\ref{eqn:upperbound_var}) can be found in \cite{Papadatos}.

Note that (\ref{eqn:pollac_khin}) strictly increases as either $\expec[S]$ or $\var[S]$ increases. Thus, we can substitute the upper bounds in it to obtain the upper bound on mean response time (\ref{eqn:upper_bnd_gen}).

Finally, we note that our proof in Theorem~\ref{thm:lower_bnd} cannot be extended to this general service time setting. The proof requires memoryless property of the service time, which does not necessary hold in the general service time case.
\end{proof}

%--------- Yanpei Add Ended ----------------------------------------------

\subsection{Numerical Examples and Simulation}
\label{subsec:fj_results}

In this section we present numerical and simulation example results to help us appreciate how storing the content on
$n$ disks using an $(n,k)$
fork-join system (as described in Section~\ref{subsec:sys_model}) reduces the expected download time.
The results demonstrate the tightness of the bounds derived in Section~\ref{subsec:bounds}.
In addition, we simulate the fork-join system to obtain an empirical cumulative density function (CDF) for the download time.

The download time of a file with $k$ blocks can be improved by increasing 1) the storage expansion $n/k$ per file
and/or 2) the number $n$ of disks used for file storage. For example, fork-join systems $(4,2)$ and $(10,5)$ both provide
a storage expansion of $2$, but the former uses $4$ and the latter $10$ disks, and thus their download times behave differently.
Both the total storage and the number of storing elements could be a limiting factor in practice.

We first address the scenario where the number of disks disks $n$ is kept constant, but the storage expansion changes
from $1$ to $n$ as we choose $k$ from $n$ to $1$. We then study the scenario where the storage expansion factor $n/k$ is kept
constant, but the number of disks varies.

\subsubsection{Flexible Storage Expansion \& Fixed Number of Disks}

In Fig.~\ref{fig:tightness_bounds} we plot the mean response time $T_{(n,k)}$ versus $k$ for fixed number of disks $n=10$, arrival rate
$\lambda = 1$ request per second and service rate $\mu = 3$ units of data per second. Each disk stores $1/k$ units of data and thus
the service rate of each individual queue is $\mu' = k \mu$.
\begin{figure}[hbt]
\includegraphics[scale=0.55]{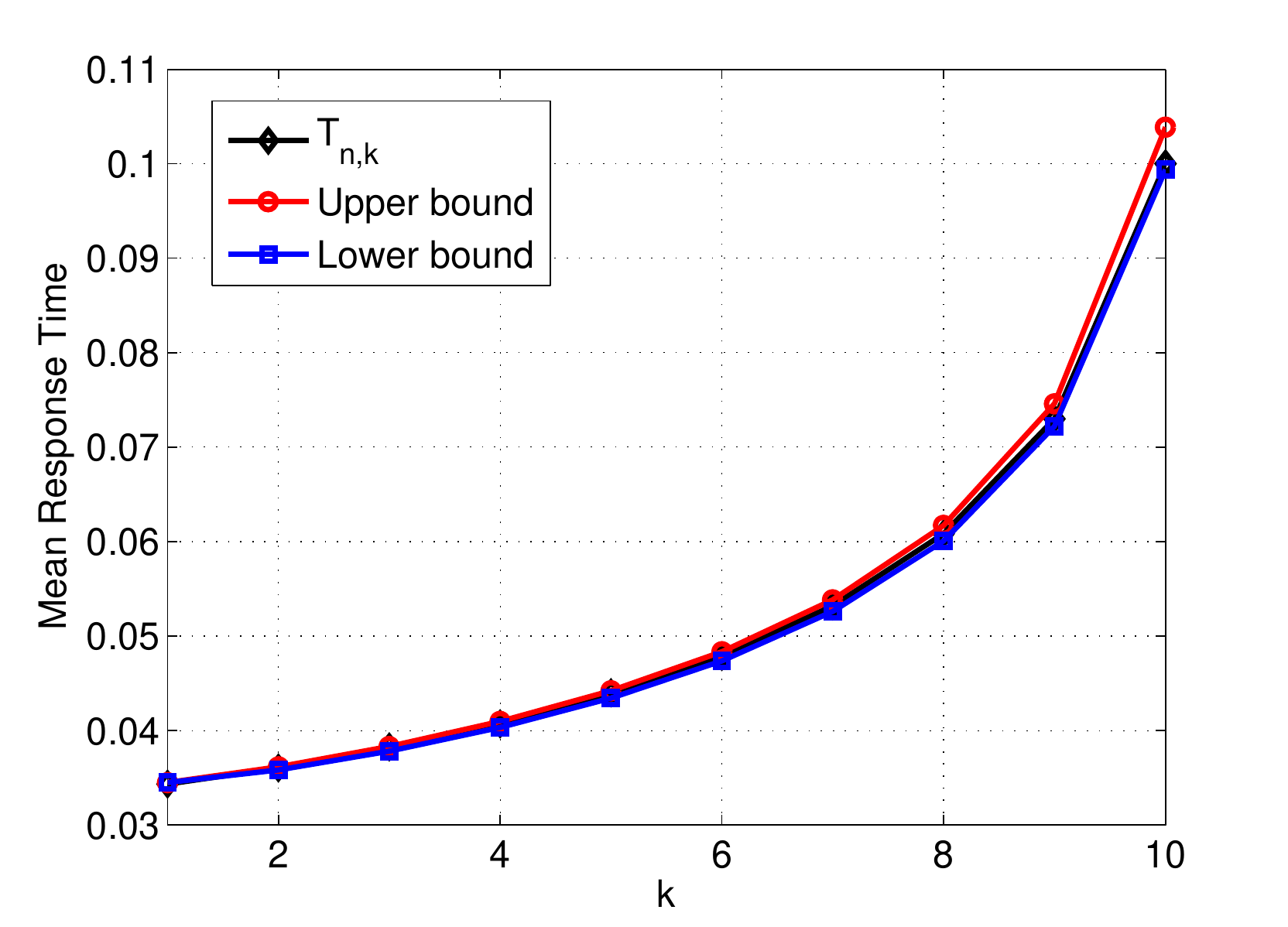}
\caption{Mean response time $T_{(n,k)}$ increases with $k$ for fixed $n$ because the redundancy of coding reduces.
The plot also demonstrates the tightness of the bounds derived in Section~\ref{subsec:bounds}\label{fig:tightness_bounds}}
\end{figure}
The simulation plot shows that as $k$ increases with $n$ fixed, the code
rate $k/n$ increases thus reducing the amount of redundancy. Hence, $T_{(n,k)}$ increases with $k$. We also observe that the bounds (\ref{eqn:upper_bnd}) and (\ref{eqn:lower_bnd}) derived in Section~\ref{subsec:bounds} are very tight.

In addition to low mean response time, ensuring quality-of-service to the user may also require that the probability of exceeding some maximum tolerable response time is small. Thus, we study the CDF of the response time for different values of $k$ for a fixed $n$.

In Fig.~\ref{fig:flexible_storage_cdf} we plot the CDF of the response time with $k= 1, 2 , 5, 10$ for fixed $n=10$. The arrival rate and service rate are $\lambda = 1$ and $\mu = 3$ as defined earlier. For $k=1$, the PDF is represents the minimum of $n$ exponential random variables, which is also exponentially distributed.
%The peak of the PDF shifts to the right as $k$ increases.
\begin{figure}[hbt]
%\begin{tikzpicture}
%\node at (0,0) (P) {\includegraphics[scale=0.5]{fs-fj-cdf}};
%\end{tikzpicture}
%
\begin{tikzpicture}
\node at (0,0) (P) {\includegraphics[scale=0.5]{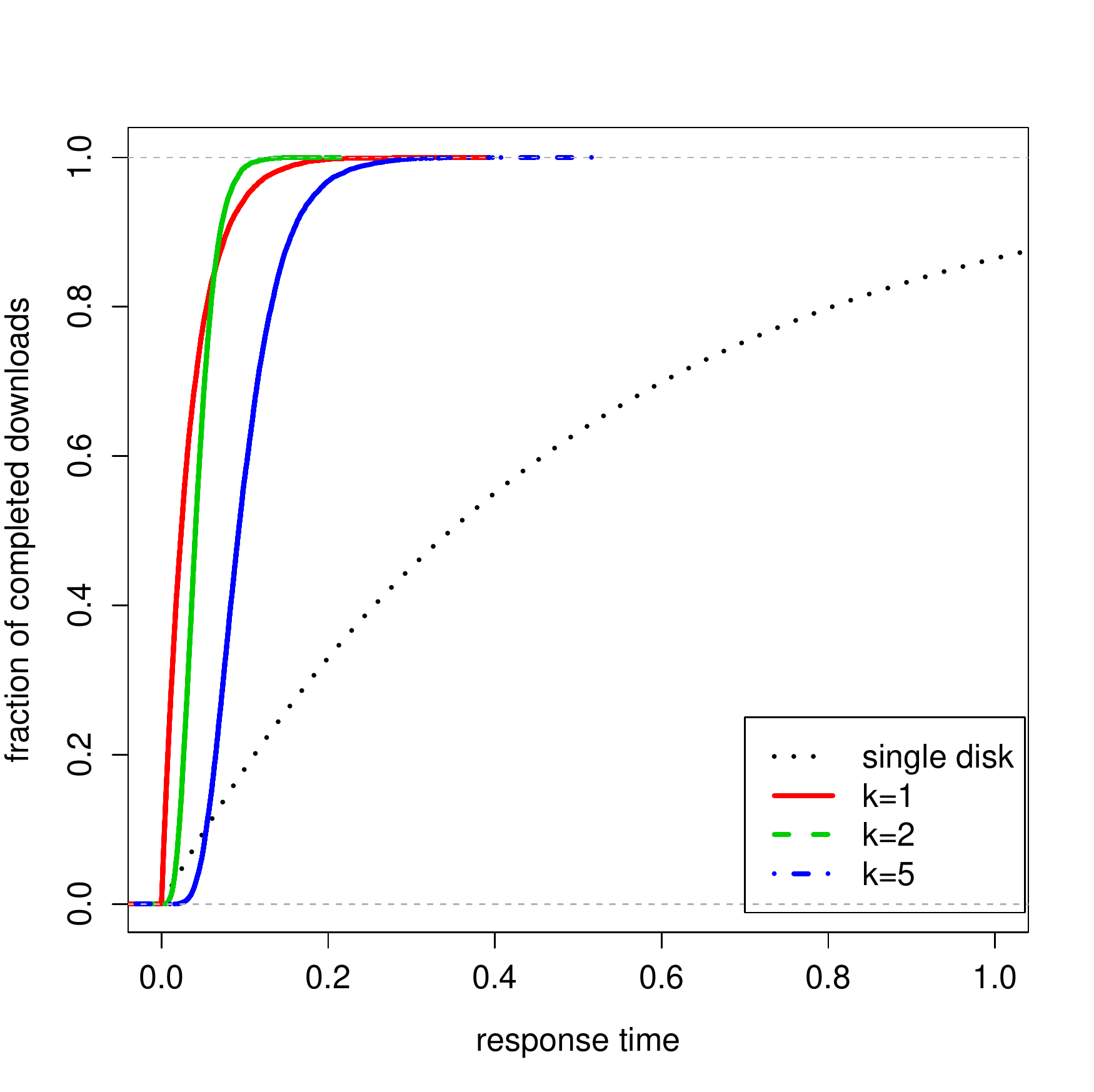}};
\end{tikzpicture}
\\[1mm]
\begin{tikzpicture} [scale=0.35]
\def\myangle{360};
\def\dr{25pt}
\def\hr{3pt}
\def\tr{5pt}
\foreach \n in {0} {
\coordinate (A) at (\n,0);
\draw [black!30, - , fill] (A) -- ($(A)+(\dr,0)$)  arc (0:360:\dr) -- cycle;
\draw [black, - , fill] (A) -- ($(A)+(\dr,0)$)  arc (0:\myangle:\dr) -- cycle;
\draw [white, fill] (A) circle (\tr);
\draw [black!50] (A) circle (\hr);
}
\draw (1.5*\dr,0) node [right] {$\leftarrow$ \small{unit storage requirement per file}};
\end{tikzpicture}
\\[1mm]
\begin{tikzpicture} [scale=0.35]
\def\myangle{36};
\def\dr{25pt}
\def\hr{3pt}
\def\tr{5pt}
\foreach \n in {0,2,4,6,8,10,12,14,16,18} {
\coordinate (A) at (\n,0);
\draw [black!30, - , fill] (A) -- ($(A)+(\dr,0)$)  arc (0:360:\dr) -- cycle;
\draw [blue, - , fill] (A) -- ($(A)+(\dr,0)$)  arc (0:\myangle:\dr) -- cycle;
\draw [white, fill] (A) circle (\tr);
\draw [black!50] (A) circle (\hr);
}
\draw ($(A)+(1.5*\dr,0)$) node [right] {\small{$n/k=1$}};
\end{tikzpicture}
\\[1mm]
\begin{tikzpicture} [scale=0.35]
\def\myangle{72};
\def\dr{25pt}
\def\hr{3pt}
\def\tr{5pt}
\foreach \n in {0,2,4,6,8,10,12,14,16,18} {
\coordinate (A) at (\n,0);
\draw [black!30, - , fill] (A) -- ($(A)+(\dr,0)$)  arc (0:360:\dr) -- cycle;
\draw [green, - , fill] (A) -- ($(A)+(\dr,0)$)  arc (0:\myangle:\dr) -- cycle;
\draw [white, fill] (A) circle (\tr);
\draw [black!50] (A) circle (\hr);
}
\draw ($(A)+(1.5*\dr,0)$) node [right] {\small{$n/k=2$}};
\end{tikzpicture}
\\[1mm]
\begin{tikzpicture} [scale=0.35]
\def\myangle{360};
\def\dr{25pt}
\def\hr{3pt}
\def\tr{5pt}
\foreach \n in {0,2,4,6,8,10,12,14,16,18} {
\coordinate (A) at (\n,0);
\draw [black!30, - , fill] (A) -- ($(A)+(\dr,0)$)  arc (0:360:\dr) -- cycle;
\draw [red, - , fill] (A) -- ($(A)+(\dr,0)$)  arc (0:\myangle:\dr) -- cycle;
\draw [white, fill] (A) circle (\tr);
\draw [black!50] (A) circle (\hr);
}
\draw ($(A)+(1.5*\dr,0)$)  node [right] {\small{$n/k=10$}};
\end{tikzpicture}
\caption{CDFs of the response time of $(10, k)$ fork-join systems, and the required storage\label{fig:flexible_storage_cdf}}
\end{figure}
%\end{example}

The CDF plot can be used to design a storage system that gives probabilistic bounds on the response time. For example,
if we wish to keep the response time below $0.1$ seconds with probability at least $0.75$, then the CDF plot shows that $k=5,10$
satisfy this requirement but $k=1$ does not. The plot also shows that at $0.4$ seconds, $100\%$ of requests are complete
in all fork-join systems, but only $50\%$ are complete in the single-disk case

\subsubsection{Flexible Number of Disks \& Fixed Storage Expansion }
Now we take a different viewpoint and analyze the benefit of spreading the content across more disks while using the same total storage space.
Fig.~\ref{fig:T_n_k_with_const_ratio} shows a simulation plot of the mean response time $T_{(n,k)}$ versus $k$ while keeping constant code rate
$k/n= 1/2$.  The response time $T_{(n,k)}$ reduces with increase in $k$ because we get the diversity advantage of having more disks.
\begin{figure}[hbt]
\includegraphics[scale=0.55]{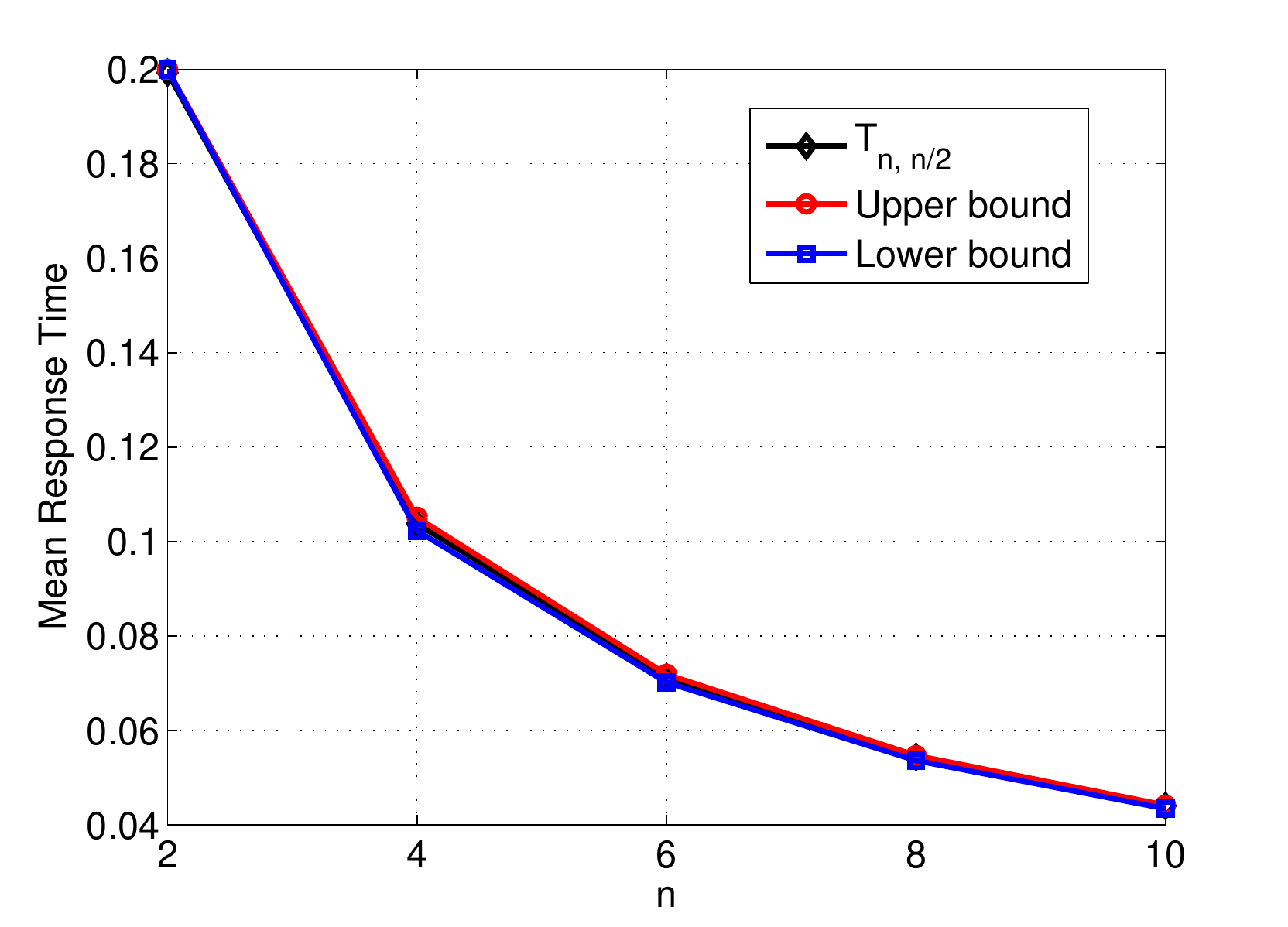}
\caption{The mean response time with constant code rate  \label{fig:T_n_k_with_const_ratio}}
\end{figure}
With a very large $n$, as $k$ increases, the theoretical bounds (\ref{eqn:upper_bnd}) and (\ref{eqn:lower_bnd}) suggest that $T_{(n,k)}$ approaches zero.
This is because we assumed that service rate of a single disk $\mu' = k \mu$ since the $1/k$ units of the content $F$ is stored on one disk.
However, in practice the mean service time $1/\mu'$ will not go zero because reading each disk will need some non-zero processing delay
in completing each task irrespective of the amount of data stored on it.

In order to understand the response time better, we plot its CDF in Fig.~\ref{fig:fixed_storage_cdf} for different values of $k$ for fixed ratio $
k/n = 1/2$. Again we observe that the diversity of increasing number of disks $n$ helps reduce the response time.
\begin{figure}[hbt]
\begin{tikzpicture}
\node at (0,0) (P) {\includegraphics[scale=0.5]{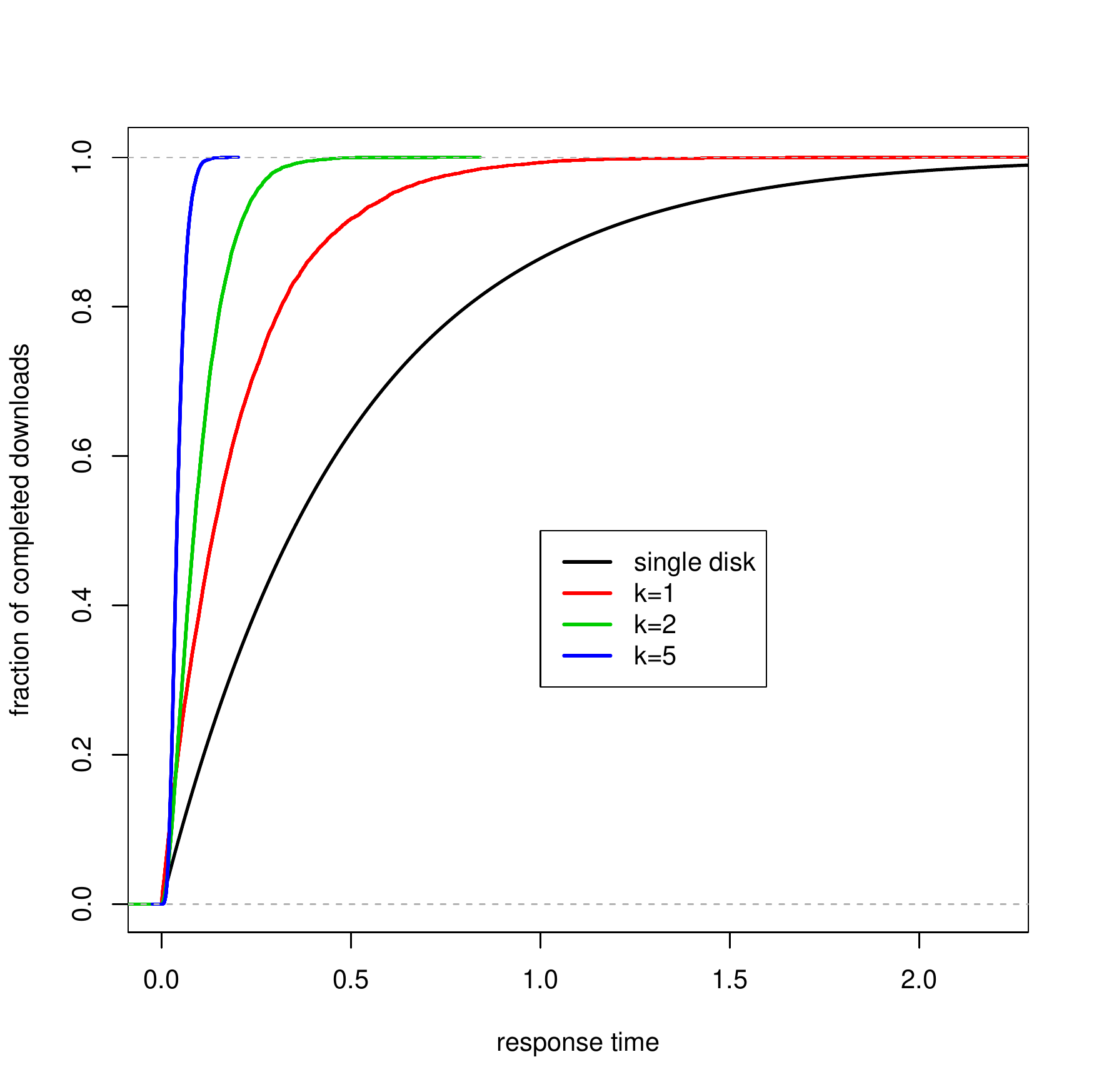}};
\end{tikzpicture}
\begin{tikzpicture} [scale=0.35]
\def\myangle{360};
\def\dr{25pt}
\def\hr{3pt}
\def\tr{5pt}
\foreach \n in {0} {
\coordinate (A) at (\n,0);
\draw [black!30, - , fill] (A) -- ($(A)+(\dr,0)$)  arc (0:360:\dr) -- cycle;
\draw [black, - , fill] (A) -- ($(A)+(\dr,0)$)  arc (0:\myangle:\dr) -- cycle;
\draw [white, fill] (A) circle (\tr);
\draw [black!50] (A) circle (\hr);
}
\end{tikzpicture}
\\[1mm]
\begin{tikzpicture} [scale=0.35]
\def\myangle{360};
\def\dr{25pt}
\def\hr{3pt}
\def\tr{5pt}
\foreach \n in {0,2} {
\coordinate (A) at (\n,0);
\draw [black!30, - , fill] (A) -- ($(A)+(\dr,0)$)  arc (0:360:\dr) -- cycle;
\draw [red, - , fill] (A) -- ($(A)+(\dr,0)$)  arc (0:\myangle:\dr) -- cycle;
\draw [white, fill] (A) circle (\tr);
\draw [black!50] (A) circle (\hr);
}
\end{tikzpicture}
\\[1mm]
\begin{tikzpicture} [scale=0.35]
\def\myangle{180};
\def\dr{25pt}
\def\hr{3pt}
\def\tr{5pt}
\foreach \n in {0,2,4,6} {
\coordinate (A) at (\n,0);
\draw [black!30, - , fill] (A) -- ($(A)+(\dr,0)$)  arc (0:360:\dr) -- cycle;
\draw [green, - , fill] (A) -- ($(A)+(\dr,0)$)  arc (0:\myangle:\dr) -- cycle;
\draw [white, fill] (A) circle (\tr);
\draw [black!50] (A) circle (\hr);
}
\end{tikzpicture}
\\[1mm]
\begin{tikzpicture} [scale=0.35]
\def\myangle{72};
\def\dr{25pt}
\def\hr{3pt}
\def\tr{5pt}
\foreach \n in {0,2,4,6,8,10,12,14,16,18} {
\coordinate (A) at (\n,0);
\draw [black!30, - , fill] (A) -- ($(A)+(\dr,0)$)  arc (0:360:\dr) -- cycle;
\draw [blue, - , fill] (A) -- ($(A)+(\dr,0)$)  arc (0:\myangle:\dr) -- cycle;
\draw [white, fill] (A) circle (\tr);
\draw [black!50] (A) circle (\hr);
}
\end{tikzpicture}
\
\caption{CDFs of the response time of $(2k, k)$ fork-join systems, and the required storage\label{fig:fixed_storage_cdf}}
\end{figure}

% -------END ----- Added by Gauri - 20th July --------------------------

%$10$-Fork, $k$-Join 

\section{Conclusion and Future Work}
\label{sec.conclusion}

We analyzed the download time of a content file from a distributed storage system. 
We assume that content of interest is available redundantly on multiple disks, or on multiple nodes throughout the network, entirely or in chunks. Such scenarios may be a consequence of caching throughout a network or as a result of purposeful storage in data centers and storage area networks. Our idea is to also make redundant requests for content in order to reduce the download time through route diversity (the fountain model)
and load balancing (the queuing model). Under this central idea, we showed that the expected download time is significantly reduced using coding -- we divide the content into $k$ parts, apply an $(n,k)$ MDS code and store it on $n$ disks. The file can be recovered by reading any $k$ of the $n$ disks. We analytically studied the mean response time and derived tight upper and lower bounds.

In practical storage systems, adding redundancy in storage not only requires extra capital investment in storage devices, networking and management but also consumes more energy. It has been estimated that around $40\%$ of total operation cost has been related to power distribution, cooling, and electricity bills \cite{CostOfCloud} and the total data center power consumption in 2005 was already $1\%$ of the total US power consumption \cite{MathewSitaramanShenoy}. It would be interesting to study the fundamental tradeoff between power consumption and quality of service (QoS) performance and distill insight on system design. As we have shown in this paper, for the same performance, coding requires less redundancy than conventional replication based storage. We would like to investigate how much energy  can coding based storage save us. Furthermore, this also motivates the research on more efficient load balancing algorithms, which not only fork each job onto a set of servers, but do so with power conservation in mind.

In this paper we do not consider some other possible costs, such as the decoding time required to reconstruct the original content out of $k$ received blocks, or placing redundant requests.  We try to qualitativly illustrate the possible benefits of coding without exactly quantifying the gains in particular, practical systems. Taking the decoding time (which affects delay performance) into consideration motivates us to investigate the optimal redundancy level. We also leave this as our future work.

%$10$-Fork, $k$-Join

%========================== Bibliography  ========================%
\bibliographystyle{IEEEtran}
\bibliography{fork_join}

\end{document}